\documentclass[journal]{ieeetran}
\usepackage{graphicx}
\usepackage{subcaption}
\usepackage{cite}
\usepackage{color, flushend}
\usepackage{booktabs,url}
\usepackage{amsmath,amssymb,amsthm,acronym}
\usepackage{algorithm,algorithmic,multicol}

\newtheorem{proposition}{Proposition}{}

\title{Time-Domain Channel Estimation for Extremely Large MIMO THz Communication Systems Under Dual-Wideband Fading Conditions}

\author{Evangelos Vlachos,~\IEEEmembership{Member,~IEEE}, Aryan Kaushik,~\IEEEmembership{Member,~IEEE}, Yonina C. Eldar,~\IEEEmembership{Fellow,~IEEE},\\ and George C. Alexandropoulos,~\IEEEmembership{Senior~Member,~IEEE}
\thanks{This work has been supported by the Smart Networks and Services Joint Undertaking (SNS JU) project TERRAMETA under the European Union’s Horizon Europe research and innovation programme under Grant Agreement No 101097101, including top-up funding by UK Research and Innovation
(UKRI) under the UK government's Horizon Europe funding guarantee.}
\thanks{E. Vlachos is the Industrial Systems Institute, ATHENA Research and Innovation Centre, 26504 Rio-Patras, Greece (e-mail: evlachos@athenarc.gr).}
\thanks{A. Kaushik is with the School of Engineering and Informatics, University of Sussex, UK (e-mail: aryan.kaushik@sussex.ac.uk).}
\thanks{Y. C. Eldar is with the Faculty of Mathematics and Computer Science, Weizmann institute of Science, Israel (e-mail: yonina.eldar@weizmann.ac.il).}
\thanks{G. C. Alexandropoulos is with the Department of Informatics and Telecommunications, National and Kapodistrian University of Athens, 15784 Athens, Greece and with the Department of Electrical and Computer Engineering Department, University of Illinois Chicago, Chicago, IL 60601, USA (e-mail: alexandg@di.uoa.gr).}
}

\begin{document}
\maketitle
\begin{abstract}
In this paper, we study the problem of extremely large (XL) multiple-input multiple-output (MIMO) channel estimation in the terahertz (THz) frequency band, considering the presence of propagation delays across the entire array apertures at both communication ends, which naturally leads to frequency selectivity. This problem is known as beam squint and may be pronounced when communications are subject to multipath fading conditions. Multi-carrier (MC) transmission schemes, which are usually deployed in THz communication systems to address these issues, suffer from high peak-to-average power ratio, which is specifically dominant in this frequency band where low transmit power is mostly feasible. Furthermore, the frequency selectivity caused by severe molecular absorption in the THz band necessitates delicate consideration in MC system design. Motivated by the benefits of single-carrier (SC) waveforms for practical THz communication systems, diverging from the current dominant research trend on MC systems, we devise a novel channel estimation problem formulation in the time domain for SC XL MIMO systems subject to multipath signal propagation, spatial wideband effects, and molecular absorption. An efficient alternating minimization approach is presented to solve the proposed mixed-integer sparse problem formulation. The conducted extensive performance evaluation results validate that the proposed XL MIMO estimation scheme exhibits superior performance than conventional SC- and MC-based techniques, approaching the idealized lower bound.
\end{abstract}


\section{Introduction}\label{introduction}
Terahertz (THz) communications (in the range of $0.1-10$ THz) have recently received remarkable attention within the global wireless community due to their increased potential for seemless data transfer, wide bandwidth that can theoretically reach up to hundreds of gigahertz (GHz), data rates of the order of terabits per second enabling ultra–fast downloading for immersive applications, and latency in the order of microseconds ($\mu$sec)~\cite{9794668}. Therefore, it has been recently recognised as one of the promising candidate technology for future sixth generation (6G) wireless networks~\cite{6gvision}. 

To confront with the high penetration loss at the THz frequency band, extremely large (XL) multiple-input multiple-output (MIMO) are being considered~\cite{XLMIMO_tutorial}, capable of realizing highly directive beamforming. However, due to the ultra-high bandwidth employed in THz communications, the propagation delay across the large antenna arrays at the communication terminals can exceed the sampling period. This spatial-wideband effect causes the so-called \textit{beam squint} in the frequency domain, according to which the angle-of-arrival (AoA) varies with frequency, and consequently, the array gain becomes frequency selective~\cite{8882325}. Additionally, certain frequency ranges within the THz band suffer from severe molecular absorption loss, according to which the wave energy within the propagation medium converts into internal kinetic energy of molecules. This phenomenon further contributes to path loss and frequency selectivity in THz communications~\cite{5995306}.

\subsection{Literature Review}

The predominant literature in channel estimation with a specific focus on the beam-squint effect revolves around schemes relying on Orthogonal Frequency Division Multiplexing (OFDM)~\cite{8882325, 9049103, 9351751, xu_overcoming_2023}. However, this modulation scheme often grapples with the challenge of a high peak-to-average power ratio (PAPR), a predicament exacerbated in the context of ultra-high-frequency transmissions in the THz range where mainly low transmit power levels are feasible up to date. Furthermore, THz-specific channel-induced impairments and the presence of phase noise have been lately documented, posing additional hurdles for multi-carrier (MC) transmission strategies. On the other hand, single-carrier (SC) waveforms are known for having lower PAPRs compared to OFDM which makes them robust to system impairments and phase noise. Especially for THz communications, due to the low output power and the non-linearity effect induced by the available THz power amplifiers (PAs), it is preferable to use SC transmissions rather than OFDM~\cite{9866847, 8645479}. Moreover, the first sub-THz standard (IEEE 802.15.3d \cite{8066476}) describes an SC modulation mode to support long range and high data rate wireless applications (such as $100$ Gbps).

For millimeter wave (mmWave) communications, where the spatial-wideband effect also is present, there have appeared numerous published articles relaying on OFDM transmissions. In~\cite{8882325}, the authors presented a parametric channel estimator where the delay and angles were derived for all the subcarriers. While this approach simplifies the estimation problem, careful selection of parameters is crucial to avoid ill-conditioning and potential divergence from the solution. In~\cite{9049103}, the dual-wideband effect was addressed for mmWave MIMO OFDM for parametric channel estimation. To avoid unstable initialization, the approach relied on tensor-based modeling and problem decomposition. A scenario with single-antenna multiple users was addressed in~\cite{9351751}, where the channel parameters were obtained via the maximum a posteriori criterion. 

An adaptive deep learning approach for THz XL MIMO channel estimation was proposed in ~\cite{10143629}. To mitigate wideband effects, the authors suggest extending their technique by employing parallel streams for each subcarrier, leveraging the learned codebooks. However, given that THz systems will have large bandwidths, and thus, many subcarriers, this approach could lead to very high computational complexity. In \cite{peng_precoding_2019}, focusing on a multi-user scenario with single antennas and SC transmissions, a joint precoding and signal detection technique was proposed that capitalized on the sparsity property of THz channels and utilized the least-square QR algorithm. In \cite{8645479}, an SC sparsity-based algorithm was proposed for indoor THz channel estimation, which incorporated high molecular absorption, but did not consider XL MIMO systems neither the spatial wideband effect. 

To extend the communication distance and efficiency in THz wireless systems, several advanced architectures have been lately proposed, including adaptive designs at the physical layer~\cite{8778669, FD_HMIMO_2023, 9205899} and time reversal~\cite{timereversal6g}, XL MIMO \cite{9216613,10045774,9324910}, and multi-functional RISs \cite{hardware2020icassp,HRIS_Mag,CE_overview_2022,amplifying_RIS_2022,singh2023indexed, pimrc2023,10119089}. At this frequency band, due to the small form factor and inter-element spacing, it is feasible to design XL antenna arrays with very large numbers of antenna elements, which enables highly directive beamforming that can combat the high propagation loss. Typically, in UPAs, hundreds, or even thousands, of densely packed antennas are being considered~\cite{9794668,9324910}. These systems are of special interest since they can effectively increase communication range, thus, further enhance capacity, as well as angular resolution in THz wireless networks. On another direction, RISs have recently emerged as a promising new paradigm to achieve smart and reconfigurable wireless propagation environments~\cite{Strinati2021Reconfigurable,Alexandropoulos2022Pervasive,RIS_challenges}, and their XL versions are lately being studied for THz communications and sensing~\cite{TERRAMETA_website}.

\subsection{Motivation and Contributions}
Although the problems of beam squint and channel estimation have been extensively studied for MC systems in THz communications, the investigation of the former on SC MIMO transmissions has been significantly limited. As previously mentioned, the low output power and the non-linearity induced by available THz PAs to date motivates the adoption of SC transmissions rather than OFDM \cite{9866847}. To this end, \cite{kim_spatial_2021} presented a channel estimation technique for an SC mmWave system with single-antenna users. In this paper, we focus on the general estimation problem of XL channel matrices in SC point-to-point THz MIMO systems subject to the beam-squint effect. Our contributions are summarized as follows:
\begin{itemize}
\item We present a novel \textit{time-domain model} for the recovery of the structured channel matrix under THz XL MIMO communications. The beam-squint effect for both the transmitter (TX) and receiver (RX) is modeled along with the propagation path delays which introduce inter-symbol interference. It is noted that prior SC studies (e.g., \cite{9399122}) are usually assuming single-antenna transmitters, hence, ignoring the double-sided effect appearing in symmetric MIMO systems, which is even more pronounced when considered with multipath signal propagation conditions.
\item The effect of molecular absorption further amplifies the overall fading and the system's frequency selectivity at THz frequencies \cite{8123513}. This phenomenon in conjunction with multipath propagation and beam squint at both the TX and RX have not been previously considered in the context of XL MIMO THz communications. These factors collectively create \textit{dual-wideband fading conditions}. Our novel SC modeling approach facilitates the unified treatment of frequency selectivity arising from all these factors by employing multi-tap filtering at the RX.
\item We introduce a novel \textit{mixed-integer sparse problem} formulation that effectively incorporates the dual-wideband effects into channel estimation. The proposed formulation readily accommodates the application of efficient optimization techniques, and in this case, we adopt the alternating direction method of multipliers (ADMM).
\end{itemize}

The performance of the proposed channel estimation approach is investigated via extensive simulation results for varying system and channel parameters, as well as through comparisons with benchmark techniques that are available in the open technical literature.

\begin{table}
  \centering
  \caption{The Mathematical Symbols used in this Paper.}
  \begin{tabular}{r|l}
    \toprule
    $a, \mathbf{a}$, and $\mathbf{A}$ &  Scalar, vector, and matrix \\
    $(\cdot)^*$ & The complex conjugate of the input \\
    $j \triangleq \sqrt{-1}$ & The imaginary unit \\
    $\lceil \cdot \rceil$ & Smallest integer greater/equal of the input \\
    $\langle \mathbf{X}, \mathbf{Y} \rangle$ & Indicates the operation $\mathbf{X}^{\rm H} \mathbf{Y} + \mathbf{Y}^{\rm H} \mathbf{X}$ \\
    $\mathbf{A}^{\rm T}$, $\mathbf{A}^{\rm H}$, and $\mathbf{A}^{-1}$&  Matrix transpose, Hermitian and inverse transpose \\
    $[\mathbf{A}]_{i,j}$ &  Matrix $\mathbf{A}$ element at the $i$-th row \\ & and $j$-th column \\
    $[\mathbf{a}]_{i}$ &  The $i$-th element of vector $\mathbf{a}$  \\
    $\mathbf{I}_N$ & $N \times N$ identity matrix \\
    $\mathbf{0}_{N \times M}$ & $N \times M$ matrix with zeros\\
    $\mathbf{1}_{M \times N}$ & $M \times N$ matrix containing only $1$'s \\
    $\boldsymbol{\delta}_i$ & Aa vector with zeros and only one unity at $i$-th row \\
    $\Vert \mathbf{A} \Vert_{\rm F}$ & Matrix Frobenius norm $\sqrt{\mathbf{A}^{\rm H} \mathbf{A}}$ \\
    $\Vert \mathbf{x} \Vert_0$ & Pseudo-norm that counts the non-zero entries \\ 
    $\times$, $\circ$ and $\otimes$ & Scalar, Hadamard, and Kronecker products \\
    ${\rm blkdiag}(\mathbf{A}_1, \mathbf{A}_2, \ldots)$ & Indicates the operation $\sum_i \mathbf{e}_i \mathbf{e}_i^{\rm T} \otimes \mathbf{A}_i$ \\ & where $\mathbf{e}_i$'s form the canonical bases of $\mathbb{R}$ \\
    $\mathcal{F}$ & The set of constant-modulus complex numbers  \\
    $\mathbb{R}, \mathbb{C}$ & The sets of real and complex numbers\\
    $x\sim \mathcal{CN}(0,\sigma^2)$ & $x$ is a zero-mean complex Gaussian random\\ & variable with variance $\sigma^2$\\
    $\mathcal{E}\{x\}$ & Expectation of random variable $x$ \\
    \bottomrule
  \end{tabular}
  \label{table:notation}
\end{table}
\subsection{Notation and Organization} 
A summary of the notation used throughout this paper is included in Table~\ref{table:notation}, while its remainder is organized as follows: Section II presents the considered channel and system models, while Section III includes the proposed problem formulation and the position-aided channel estimation approach. Section IV verifies the proposed estimation framework through simulation results. Section V sums up the outcomes of the proposed framework and sketches directions for future research.

\section{System and Channel Models}\label{system_model}
We consider a point-to-point THz MIMO communication system comprising an $N$-antenna base station (BS) and an $M$-antenna user equipment (UE). The antenna elements at both nodes are structured in ULAs and, for their data communication, SC transmissions are adopted over a designated carrier frequency $f_c$ resulting in wavelength $\lambda_c=c/f_c$, a bandwidth $W$, and a corresponding sampling period $T_s=1/(2W)$.

\subsection{Wideband Channel Model}
While the prevalent trend in THz channel modeling literature typically involves LoS-dominant scenarios~\cite{9514889}, it is essential to recognize that multipath wideband channels may emerge in various THz scenarios~\cite{9794668}, especially within smart wireless environments featured by reconfigurable metasurfaces~\cite{THz_RIS}. To model multipath THz $M\times N$ MIMO wireless channels, we adopt the Saleh-Valenzuela channel model that is based on the time-cluster spatial-lobe approach~\cite{1146527}. To this end, the $t$-th time instance channel between each $n$-th transmit and each $m$-th receive antennas, with $t=1,2,\ldots, T$, $n=1,2,\ldots,N$ and $m=1,2,\ldots,M$, is given by the discrete time baseband model~\cite{Tse_Viswanath_2005}:
\begin{equation}\label{eq:channel_mn_t}
    h_{m,n}(t) = \sum_{\ell=1}^{L_p} a_\ell e^{-j2\pi d_{m,n,\ell}/\lambda_c} \text{sinc}(t -\tau_{m,n,\ell} ),
\end{equation}
where ${L_p}$ is the number of resolvable propagation paths due to multipath propagation; $\tau_{m,n,\ell}$ is the propagation time delay of the $\ell$-th channel path ($\ell=1,\ldots,{L_p}$) between the $m$-th receiving and $n$-th transmitting antennas; $d_{m,n,\ell}$ is the distance between these antennas for the $\ell$-th propagation path, and $\alpha_\ell$ is the instantaneous channel gain, modeled as $\alpha_\ell \sim \mathcal{N}(0, \sigma_a^2)$.

In the context of far-field communications, where the antenna array sizes are significantly smaller than the distances between them, the total separation between the $n$-th transmit antenna and the $m$-th receive antenna can be effectively approximated. This approximation involves summing the distance from the first transmit antenna to the first receive antenna, denoted as $d_\text{TX-RX}$, and the additional distance attributed to the signal's travel on the aperture, denoted as $d_{m,n,\ell}$. Specifically, this approximation is expressed as follows:
\begin{equation}
    d_\text{TOTAL} = d_\text{TX-RX} + d_{m,n,\ell}, 
\end{equation}
where we used the notation:
\begin{equation}\label{eq:d_mnl}
     d_{m,n,\ell} \triangleq (m-1) \Delta_c \cos \vartheta_{\text{TX}, \ell} - (n-1) \Delta_c \cos \vartheta_{\text{RX}, \ell}
\end{equation}
with $\Delta_c = \frac{\lambda_c}{2}=\frac{c}{2 f_c}$ denoting the antenna separation and $\vartheta_{\text{TX}, \ell}, \vartheta_{\text{RX}, \ell} \in [-\pi/2, \pi/2]$ are the physical angles of arrival and departure, respectively. Putting all above together, the MIMO channel representation in \eqref{eq:channel_mn_t} becomes:
\begin{equation}\label{eq:channel_mn_t2}
    h_{m,n}(t) = \sum_{\ell=1}^{L_p} \bar{a}_\ell \underbrace{e^{-j2\pi (m-1) \theta_{\text{RX}, \ell}} e^{j2\pi (n-1) \theta_{\text{TX}, \ell}}}_{\triangleq c_{m,n,\ell}} \text{sinc}(t - \tau_{m,n,\ell})
\end{equation}
with $\bar{\alpha}_\ell \triangleq \alpha_\ell e^{-j2\pi d_\text{TX-RX}/\lambda_c}$, $\theta_{\text{RX}, \ell} \triangleq \Delta_c \cos(\vartheta_{\text{RX}, \ell})$ is the normalized AoA and $\theta_{\text{TX}, \ell} \triangleq \Delta_c \cos(\vartheta_{\text{TX}, \ell})$ is the normalized angle of departure (AoD).

In THz channel models, the variance of each $\ell$-th channel gain coefficient depends on the respective propagation distance $d_\text{TX-RX}$ between the TX and RX as well as the carrier frequency $f_c$ via the following expression:
\begin{equation}\label{eq:channel_gain_los}
    \sigma_a^2(f_c) \triangleq \sqrt{\frac{N M}{{L_p}}} \frac{1}{d_{\text{TX-RX}}^{\xi_\ell}} e^{-\frac{1}{2} \mathcal{K}(f_c)},
\end{equation}
where $\xi_\ell$ represents the pathloss exponent which is equal to $\xi_1=2$ for the LoS path, i.e., $\ell=1$, and $\xi_\ell=3$ for $\ell=2,\ldots,{L_p}$. $\mathcal{K}(f)$ represents the function of the molecular absorption losses that depends on the carrier frequency \cite{5995306}.

\subsection{Received Signal Model}
The considered point-to-point MIMO communication takes place on a frame-by-frame basis, where the wireless channel remains constant during each frame but may change independently from one frame to another. Every frame consists of $T$ time slots, with $t=1,2,\ldots,T$, dedicated for channel estimation, whereas the rest of the frame is used for data communication. To estimate the intended THz MIMO channel matrix, which can be XL, the $M$-antenna BS utilizes training symbols for each of the $T$ slots used for channel estimation.

When the TX sends the symbol $\bar{q}_n(t) \in \mathbb{C}$ from each $n$-th antenna, the noiseless reception at the $m$-th RX antenna can be expressed as follows~\cite{Tse_Viswanath_2005}:
\begin{equation}\label{eq:channel_output}
\hat{y}_{m,n}(t) = \sum_{i=1}^{L_t} h_{m,n}(i) \bar{q}_n(t-i),
\end{equation}
where $L_t$ is the maximum number of filter taps due to the frequency selectivity of the wideband channel. Taking into account the channel model \eqref{eq:channel_mn_t2}, the received signal is given by: 
\begin{align}\label{eq:sampled_y}
\hat{y}_{m,n}(t) =  \sum_{i=1}^{L_t}  \sum_{\ell=1}^{L_p} \bar{\alpha}_\ell c_{m,n,\ell} q_{n}(t - i - \tau_{m,n,\ell}),
\end{align}
where $q_n(t - i - \tau_{m,n,\ell}) \triangleq \bar{q}_n(t-i) \text{sinc}(i-\tau_{m,n,\ell})$.

\subsection{Combined TX-RX Beam-Squint Effect}
In lower frequency ranges, in contrast to mmWave and THz, and in non-extreme MIMO systems, the carrier frequency $f_c$ does not become significantly small and the antenna index $m$ does not reach excessively large values, thus, the delay $\tau_{m,n,\ell}$ becomes negligible. However, for THz XL MIMO systems, $\tau_{m,n,\ell}$ shifts the sampling of the transmitted signal $\bar{q}_n(t)$ in \eqref{eq:sampled_y}, creating the beam-squint effect, where different RX antennas may sample different transmitted symbols $\bar{q}_n(t)$'s. This sampling shift depends on the propagation distance $d_{m,n,\ell}$, between the $n$-th TX and $m$-th RX antenna elements, along the $\ell$-th path. To this end, the aperture propagation delay time is defined as follows:
\begin{equation}
 \tau_{m,n,\ell} \triangleq d_{m,n,\ell}/c.
\end{equation}
Note that $\tau_{m,n,\ell}$ describes the combined beam-squint effect at the TX and RX. More specifically, this delay is given by the expression:
\begin{equation}
    \tau_{m,n,\ell} = \left((m-1) \frac{1}{2 f_c} \cos{\vartheta_{\text{RX}, \ell}} - (n-1) \frac{1}{2f_c} \cos{\vartheta_{\text{TX}, \ell}} \right).
\end{equation}
Since, the AoA and AoD of each $\ell$-th propagation path are bounded within $[-\pi/2, \pi/2]$, thus, $\cos\vartheta_{\text{RX},\ell}, \cos\vartheta_{\text{TX},\ell} \in [0, 1]$, the aperture delay time can be upper bounded as follows:
\begin{equation}
    \tau_{m,n,\ell} \le \frac{m-1}{2 f_c} + \frac{n-1}{2 f_c} \le \frac{M+N-2}{2 f_c}.
\end{equation}
To avoid aliasing, the sampling period needs to be chosen  to upper bound the propagation delay time, i.e., $T_s>\tau_{m,n,\ell}$. This setting also sets an upper bound for the number of antenna elements that will not be affected from to the spatial wideband effect, i.e., it must hold that $M+N \le \lceil 2 f_c T_s + 2 \rceil$.

\section{Proposed THz XL MIMO Channel Estimation}
The received signal model in \eqref{eq:sampled_y} captures the beam-squint effect along with the associated inter-symbol interference. In particular, the time delay $\tau_{m,n,\ell}$ $\forall$$m,n,\ell$ is intricately influenced by the beam squint at both the TX and RX as well as the propagation path characteristics resulting from the frequency selectivity of the wideband channel. In this section, we commence with the proposed THz XL MIMO channel estimation problem formulation for SC modulation incorporating the latter dual-wideband effect, and then, describe its efficient iterative solution. Finally, we present an initialization scheme for the proposed algorithm exploiting position information and analyze the overall complexity of the proposed channel estimation technique. 

\subsection{Problem Formulation}
Dual-wideband effects complicate the estimation of the channel impulse response. Particularly challenging is determining the propagation delays for every combination of transmit and receive antennas. Furthermore, extremely large antenna arrays significantly increase this complexity, as the number of propagation delays to be estimated grows dramatically. To address this challenge, we propose a novel decomposition approach. This approach breaks down the problem into two key components: the MIMO channel matrix, denoted by $\mathbf{H}$, and a sparse matrix, denoted by $\mathbf{E}$. This sparse matrix, $\mathbf{E}$, captures the effects of the propagation delays introduced by the dual-wideband channel.

\begin{proposition}\label{Prop:RX_symbols}
The input/output relationship for the considered $M \times N$ MIMO system over an ${L_p}$-tap multipath THz and wideband channel subject to the combined effects of maximum delay $K$ and after $T$ training instances can be expressed as: 
\begin{equation}\label{eq:Y_cap}
    \mathbf{Y} = \mathbf{H} \mathbf{\Phi} \mathbf{E} + \mathbf{N},
\end{equation}
where $\mathbf{Y} \in \mathbb{C}^{M \times T}$ denotes the matrix with all $T$ received training signals from all $M$ RX antennas in baseband and $\mathbf{H}$ represents the $M \times MNL$ effective channel matrix defined as:
\begin{equation}\label{eq:H_def}
\mathbf{H} \triangleq {\rm blkdiag}(\mathbf{h}_1^{\rm T}, \ldots, \mathbf{h}_M^{\rm T}),
\end{equation}
where $\forall$$m=1,\ldots,M$: 
$$\mathbf{h}_m \triangleq [h_{m,1,1}, \ldots,h_{m,1,{L_p}}, h_{m,2,1},\ldots, h_{m,2,{L_p}},\ldots, h_{m,N,{L_p}}].$$
The matrix $\mathbf{\Phi}$ is build using the training symbols $q_n$'s, as:
\begin{align} \label{eq:Phi_def}
\mathbf{\Phi} \triangleq [\mathbf{I}_M 
\otimes {\rm blkdiag} \left((\mathbf{I}_L \otimes \mathbf{q}_1^{\rm T}(1)), \ldots, (\mathbf{I}_L \otimes \mathbf{q}_N^{\rm T}(1))\right), \ldots \nonumber \\
\mathbf{I}_M 
\otimes {\rm blkdiag} \left((\mathbf{I}_L \otimes \mathbf{q}_1^{\rm T}(T)), \ldots, (\mathbf{I}_L \otimes \mathbf{q}_N^{\rm T}(T))\right) ]
\end{align}
with
$$
 \mathbf{q}_n(i) \triangleq [ q_n(i), \ldots, q_n(i-K)]^{\rm T} \in \mathbb{C}^{K \times 1}.
$$
Finally, the matrix $\mathbf{E} \in \{0,1\}^{MNLKT \times T}$ in~\eqref{eq:Y_cap} is introduced to represent the unknown time shifts and is defined as follows: 
\begin{equation}\label{eq:E_def}
\mathbf{E} \triangleq \mathbf{I}_T \otimes [e_{1,1,1}, \ldots, e_{1,1,{L_p}}, \ldots, e_{1,N,{L_p}},\ldots, e_{M,N,{L_p}}]^{\rm T},
\end{equation}
where $e_{m,n,\ell} \in \{0,1\}$ is a binary scalar quantity.
The term $\mathbf{N} \in \mathbb{C}^{M \times N}$ represents the complex AWGN matrix that is distributed as $\mathbf{N} \sim \mathcal{N}(\mathbf{0}_{M \times N}, \sigma_N^2 \mathbf{I}_M)$.
\end{proposition}

\begin{proof}
The signal $\hat{y}_{m,n,\ell}(i)$ in \eqref{eq:sampled_y} can be rewritten as follows: 
\begin{align}
    \hat{y}_{m,n,\ell}(t, i) &= \bar{\alpha}_\ell c_{m,n,\ell} q_n(t - i - \kappa_{m,n,\ell}T_s) \\
    &= \bar{\alpha}_\ell c_{m,n,\ell}  \mathbf{q}_n^{\rm T}(t,i) \mathbf{e}_{m,n,\ell},
\end{align}
where $\mathbf{e}_{m,n,\ell} \in \{0,1\}^{K \times 1}$ is a $K \times 1$ binary vector with zeros everywhere except the $\kappa_{m,n,\ell}$-th position with $\Vert \mathbf{e}_{m,n,\ell} \Vert_0\ge 1$; $\tau_{m,n,\ell} \triangleq \kappa_{m,n,\ell} T_s$  for $\kappa_{m,n,\ell} \in [0, K]$, while $K T_s$ is the maximum delay.

While the symbol vector $\mathbf{q}_n(t, i) \in \mathbb{C}^{K \times 1}$ is known at the RX, the binary vector $\mathbf{e}_{m,n,\ell}$ has to be recovered for all TX and RX antenna elements (recall that $n=1,2,\ldots, N$ and $m=1,2,\ldots, M$) as well as for all channel propagation paths $\ell=1,2,\ldots,{L_p}$. The sampled received signal for each $m$-th RX and $n$-th TX antenna pair in \eqref{eq:sampled_y} can be expressed as follows:
\begin{align}
\hat{y}_{m,n}(t) &= \sum_{i=1}^{L_t} \sum_{\ell=1}^{L_p} \bar{\alpha}_\ell c_{m,n,\ell}  \mathbf{q}_n^{\rm T}(t, i) \mathbf{e}_{m,n,\ell}(i) \nonumber \\
&= \sum_{i=1}^{L_t} \mathbf{h}_{m,n}^{\rm T} \mathbf{Q}_n(t, i) \mathbf{e}_{m,n}(i), \label{eq:y_mni}
\end{align}
where we have used the definitions $\mathbf{Q}_n(t, i) \triangleq (\mathbf{I}_{L_p} \otimes \mathbf{q}_n^{\rm T}(t,i)) \in \mathbb{C}^{{L_p} \times {L_p} K}$ and $\mathbf{e}_{m,n} \triangleq [ \mathbf{e}_{m,n,1}^{\rm T},\ldots,\mathbf{e}_{m,n,L}^{\rm T} ]^{\rm T}\in \{0,1\}^{L K \times 1}$, 
and $\mathbf{h}_{m,n} \in \mathbb{C}^{{L_p} \times 1}$ includes the vectorized values of the channel gains for all ${L_p}$ paths, which is defined as:  
$$\mathbf{h}_{m,n} \triangleq [\bar{\alpha}_1 c_{m,n,1}, \ldots, \bar{\alpha}_L c_{m,n,{L_p}}]^{\rm T} \in \mathbb{C}^{{L_p} \times 1}$$ with elements 
$h_{m,n,\ell}=\bar{\alpha}_\ell c_{m,n,\ell}$ $\forall$$\ell$.
Similarly, \eqref{eq:y_mni} can be written as:
\begin{align}
\hat{y}_{m,n}(t) = \sum_{i=1}^{L_t} \mathbf{h}_{m,n}^{\rm T} \mathbf{Q}_n(t, i) \mathbf{e}_{m,n}(i) = \mathbf{h}_{m,n}^{\rm T} \mathbf{Q}_n(t) \mathbf{e}_{m,n},
\label{eq:y_mn}
\end{align}
with $\mathbf{Q}(t) \triangleq [\mathbf{Q}(t,1) \ldots \mathbf{Q}(t, L_t)]$ and $\mathbf{e}_{m,n}(i) \triangleq [\mathbf{e}_{m,n}^{\rm T}(1) \ldots \mathbf{e}_{m,n}^{\rm T}(T_p)]^{\rm T}$.

The noiseless received signal for all TX antennas is given by the superposition $\hat{y}_m(t) =  \sum_{n=1}^N \hat{y}_{m,n}(t)$. Using \eqref{eq:y_mn}, this baseband signal can be re-expressed as follows:
\begin{align}
   \hat{y}_m(t) &= \sum_{n=1}^N  \mathbf{h}_{m,n}^{\rm T} \mathbf{Q}_n(t) \mathbf{e}_{m,n} = \mathbf{h}_m^{\rm T} \bar{\mathbf{Q}}(t) \mathbf{e}_m,
\end{align}
where, for the last expression, we have defined the following quantities: 
\begin{align}
\bar{\mathbf{Q}}(t) &\triangleq {\rm blkdiag}(\mathbf{Q}_1(t),\ldots, \mathbf{Q}_N(t))  \in \mathbb{C}^{{L_p} N \times {L_p} N K},\nonumber\\ \nonumber
\mathbf{e}_m &\triangleq [\mathbf{e}_{m,n}^{\rm T}, \ldots, \mathbf{e}_{m,N}^{\rm T}]^{\rm T} \in \{0,1\}^{{L_p} K N \times 1},\\ \nonumber
\mathbf{h}_m &\triangleq [\mathbf{h}_{m,1}^{\rm T},\ldots,\mathbf{h}_{m,N}^{\rm T}]^{\rm T} \in \mathbb{C}^{N {L_p} \times 1}.
\end{align}
Next, for each training instance $t=1,\ldots,T$, we construct the receiving vector $\mathbf{y}(t) \in \mathbb{C}^{M \times 1}$ with the received training symbols from all the $M$ RX antennas, as follows:
\begin{equation}
     \mathbf{y}(t) = \mathbf{H} (\mathbf{I}_M \otimes \bar{\mathbf{Q}}(t)) \mathbf{e} + \mathbf{n}(t),
\end{equation}
where $\mathbf{I}_M \otimes \bar{\mathbf{Q}}(t)$ is an $MNL \times MNLK$ matrix,
\begin{align}
\mathbf{H} &\triangleq {\rm blkdiag}(\mathbf{h}_1^{\rm T}, \ldots,\mathbf{h}_M^{\rm T}) \in \mathbb{C}^{M \times M N {L_p}},\nonumber\\ \nonumber
\mathbf{e} &\triangleq [\mathbf{e}_1^{\rm T},\ldots,\mathbf{e}_M^{\rm T}]^{\rm T} \in \{0,1\}^{MNLK\times 1},
\end{align}
and $\mathbf{n}(t) \sim \mathcal{N}(\mathbf{0}_M, \sigma_N^2 \mathbf{I}_M)$. By collecting the received signals from all $T$ time instances and using the matrix notations $\mathbf{\Phi} \triangleq [\mathbf{I}_M 
\otimes \bar{\mathbf{Q}}_1,\ldots, \mathbf{I}_M 
\otimes \bar{\mathbf{Q}}_T] \in \mathbb{C}^{MNL \times MNLKT}$ and $\mathbf{E} \triangleq \mathbf{I}_{T} \otimes \mathbf{e}$, the $M \times T$ matrix $\mathbf{Y}$ in \eqref{eq:Y_cap} is obtained, which completes the proof.
\end{proof}

Capitalizing Proposition~\ref{Prop:RX_symbols}, we formulate our THz XL MIMO channel estimation objective incorporating the dual-wideband effect as the following optimization problem:
\begin{align}\label{eq:BLP}
   \mathcal{OP}: \min_{\mathbf{H}, \mathbf{E}} &\Vert \mathbf{Y} - \mathbf{H} \mathbf{\Phi} \mathbf{E} 
    \Vert_{\rm F}^2 \,\,\text{s.t.}\,\, [\mathbf{E}]_{p,q} \in \{0,1\} \nonumber \\ \nonumber
    &\forall p=1,\ldots,MNLK \,\,\text{and}\,\, \forall q=1,\ldots,T,\\ \nonumber
    & \mathbf{H} \,\,\text{as in}\,\,\eqref{eq:H_def}\,\,\text{and}\,\,\mathbf{E} \,\,\text{as in}\,\,\eqref{eq:E_def}.
\end{align}
Note that, $\mathcal{OP}$ belongs to the class of mixed-integer sparse optimization problems. 

\subsection{Idealized Solution of the Decomposed Problem}
Before delving into the proposed solution, let us consider a naive approach that involves decomposing the considered problem into the following two independent subproblems that can be solved separately~\cite{adams1993mixed}:
\begin{itemize}
    \item Assuming that the $\mathbf{H}^*$ channel matrix is known, solve for the beam-squint matrix $\mathbf{E}$:
    \begin{align}
    \mathcal{OP}_1: \mathbf{E}_{\rm opt} \triangleq \arg\min_{\mathbf{E}} & \Vert \mathbf{Y} - \mathbf{H}^* \mathbf{\Phi} \mathbf{E} \Vert_{\rm F}^2 \nonumber \\
    \text{s.t. }& \mathbf{E}_{p,q} \in \{0,1\} \text{ and } \mathbf{E}\,\,\text{as in}\,\,\eqref{eq:E_def}.\nonumber
    \end{align}
    \item Assuming that $\mathbf{E}$ is known, solve $\mathbf{H}$:
    \begin{align}
    \mathcal{OP}_2: \mathbf{H}_{\rm opt} \triangleq \arg\min_{\mathbf{H}} & \Vert \mathbf{Y} - \mathbf{H} \mathbf{\Phi} \mathbf{E}^*\Vert_{\rm F}^2.\nonumber
    \nonumber \\
    \text{s.t. }& \mathbf{H}\,\,\text{as in}\,\,\eqref{eq:H_def}.\nonumber
    \end{align}
\end{itemize}
Problem $\mathcal{OP}_1$ can be addressed by relaxing the integer constraint to a box constraint and employing a sparsity-promoting norm operator. On the other hand, $\mathcal{OP}_2$ admits a closed-form solution via unconstrained least squares, yielding:
\begin{equation}
\mathbf{H} = \mathcal{P}(\mathbf{Y} (\mathbf{\Phi} \mathbf{E}^{*})^\dagger),
\end{equation}
where $\mathcal{P}(\cdot)$ imposes the block structure of \eqref{eq:H_def}. This approach for solving $\mathcal{OP}$ in a decoupled way is outlined in Algorithm~\ref{alg:proposed1}. It is noted that, in this idealized case, the initializers $\mathbf{H}^*$ and $\mathbf{E}^*$ of each problem  are perfectly known, this two-stage successive solution of the optimization problems $\mathcal{OP}_1$ and $\mathcal{OP}_2$ results into the lowest estimation bounds. However, in practical scenarios where the initializers deviate from the optimal ones, this approach overlooks the propagated errors. This may result into much lower performance (i.e., estimation accuracy) or even divergence from the optimal solution.

\begin{algorithm}[t]
\caption{$\mathcal{OP}$'s Decomposed Solution}
\label{alg:proposed1}
\begin{algorithmic}[1]
\REQUIRE $\gamma$, $\mathbf{\Phi}$, $\mathbf{Y}$, $x_\text{BS}$, $x_\text{UE}$, and $I_\text{max}$.
\ENSURE $\mathbf{H}^{(I_\text{max})}$ and $\mathbf{E}^{(I_\text{max})}$.
\STATE Compute $\mathbf{\tilde{e}}$ that solves the relaxed problem:
\begin{equation*}
    \min_{\mathbf{\tilde{e}} \in [0,1]^{MNKL \times 1}} \Vert \mathbf{\tilde{e}} \Vert_1 + \frac{1}{2} \Vert \text{vec}(\mathbf{Y}) - (\mathbf{I} \otimes \mathbf{H}^{*} \mathbf{\Phi}) \text{vec}(\mathbf{I}_T \otimes \mathbf{\tilde{e}}) \Vert_2^2.
\end{equation*}
\STATE Calculate the threshold vector
    $\mathbf{\hat{e}} = {\rm thres}(\mathbf{\tilde{e}})$.
\STATE Obtain $\mathcal{OP}_1$'s solution as $\mathbf{E} = \mathbf{I}_T \otimes \mathbf{\hat{e}}$.
\STATE Solve $\mathcal{OP}_2$ as $\mathbf{H} = \mathcal{P}(\mathbf{Y} (\mathbf{\Phi} \mathbf{E}^{*})^\dagger)$.
\end{algorithmic}
\end{algorithm}

\subsection{Exploitation of the Channel's Sparse Structure}

The cost function introduced in $\mathcal{OP}$ leverages the sparsity of the matrix $\mathbf{E}$, which encodes the time shift delays. Since the THz XL MIMO channel matrix is also known to be sparse in the beamspace domain \cite{9514889}, exploiting this sparsity is beneficial for recovering the channel matrix with fewer training symbols. Recall that the matrix definition within our channel estimation context, $\mathbf{H}$ in~\eqref{eq:H_def}, deviates from the conventional channel matrix structure commonly employed. To this end, we derive a representation of this matrix into a similar to the beamspace domain, via the following proposition.

\begin{proposition}  
The XL MIMO channel matrix $\mathbf{H}$ defined in~\eqref{eq:H_def} and appearing in $\mathcal{OP}$ can be expressed as a block sparse matrix as follows:
\begin{equation}\label{eq:Z_def}
\mathbf{H}=\mathbf{F}_1 \mathbf{Z} \mathbf{F}_2
\end{equation}
    where we have used the matrix definitions:
        \begin{align}
           \mathbf{F}_1 \! &\triangleq \! {\rm blkdiag}((\mathbf{1}_{1\times {L_p} N} \otimes \boldsymbol{\delta}_1^{\rm T}) \mathbf{F}_{\rm TX}, \ldots, (\mathbf{1}_{1\times {L_p} N} \otimes \boldsymbol{\delta}_M^{\rm T}) \mathbf{F}_{\rm TX}), \\
           \mathbf{F}_2 \! &\triangleq \! \mathbf{I}_M \otimes {\rm blkdiag(\mathbf{I}_L \otimes \mathbf{F}_{\rm TX}^{\rm H} \boldsymbol{\delta}_1}, \ldots, \mathbf{I}_L \otimes \mathbf{F}_{\rm TX}^{\rm H}\boldsymbol{\delta}_N),
        \end{align}
    where $\mathbf{F}_1 \in \mathcal{C}^{M \times M^2 N {L_p}}$, $\mathbf{F}_2 \in \mathcal{C}^{M N^2 {L_p} \times M N {L_p}}$, and $\mathbf{Z} \in \mathcal{C}^{M^2 N {L_p} \times MN^2 {L_p}}$ is the modified beamspace for the channel formulation in \eqref{eq:H_def}.
\end{proposition}
\begin{proof}
Using the definitions within the proof of Proposition 1, the component for each $\ell$-th channel propagation path between each $m$-th RX and $n$-th TX antenna can be expressed as:
\begin{equation}
h_{m,n,\ell} \triangleq \bar{\alpha}_\ell c_{m,n,\ell} = \bar{\alpha}_\ell [\mathbf{a}_\text{RX}(\ell)]_m [\mathbf{a}_\text{TX}(\ell)]_n^*.
\end{equation}
where $\mathbf{a}_\text{RX}(\ell) \in \mathbb{C}^{M \times 1}$ and $\mathbf{a}_\text{TX}(\ell) \in \mathbb{C}^{N \times 1}$ represent the RX/TX array steering vectors, which can be expressed in the beamspace domain using the discrete Fourier matrices $\mathbf{F}_\text{RX}(\ell) \in \mathbb{C}^{M \times M}$ and $\mathbf{F}_\text{TX}(\ell) \in \mathbb{C}^{N \times N}$, as follows:
\begin{align}
    \mathbf{a}_{\rm RX}(\ell) &= \mathbf{F}_{\rm RX}^{\rm H} \mathbf{z}_{\rm RX}(\ell), \\
    \mathbf{a}_{\rm TX}(\ell) &= \mathbf{F}_{\rm TX}^{\rm H} \mathbf{z}_{\rm TX}(\ell),
\end{align}
with $\mathbf{z}_{\rm RX} \in \mathcal{C}^{M \times 1}$ and $\mathbf{z}_{\rm TX} \in \mathcal{C}^{N \times 1}$ being sparse vectors. Then, each channel coefficient $h_{m,n,\ell}$ can be expressed as:
\begin{align}
    h_{m,n,\ell} &\triangleq \bar{\alpha}_\ell [\mathbf{F}_{\rm RX} \mathbf{z}_{\rm RX}(\ell)]_m [\mathbf{F}_{\rm TX} \mathbf{z}_{\rm TX}(\ell)]_n^* \\
    &= \bar{\alpha}_\ell \boldsymbol{\delta}_m^{\rm T} \mathbf{F}_{\rm RX} \mathbf{z}_{\rm RX}(\ell) \mathbf{z}_{\rm TX}^{\rm H}(\ell) \mathbf{F}_{\rm TX}^{\rm H} \boldsymbol{\delta}_n.
\end{align}
By collecting the channel elements for all ${L_p}$ propagation paths ($\ell=1,2, \ldots, {L_p}$), it can be deduced that:
\begin{equation}
    \mathbf{h}_{m,n} = (\mathbf{1}_L \otimes \boldsymbol{\delta}_m^{\rm T} \mathbf{F}_{\rm RX}) \mathbf{Z} (\mathbf{I}_L \otimes \mathbf{F}_{\rm TX}^{\rm H} \boldsymbol{\delta}_n),
\end{equation} 
where $\mathbf{Z} \triangleq \bar{\alpha}_\ell \mathbf{z}_{\rm RX}(\ell) \mathbf{z}_{\rm TX}^{\rm H}(\ell)$. Afterwards, we collect the channel elements for all TX elements, i.e., $\forall$$n=1,2,\ldots,N$, yielding the following channel vector expression:
\begin{equation}
    \mathbf{h}_m = (\mathbf{I}_N \otimes \mathbf{Z}) {\rm blkdiag}(\mathbf{I}_L \otimes \mathbf{F}_{\rm TX}^{\rm H} \boldsymbol{\delta}_1, \ldots, \mathbf{I}_L \otimes \mathbf{F}_{\rm TX}^{\rm H} \boldsymbol{\delta}_N).
\end{equation}
The latter expression describes the input vectors at the right-hand side of the $\mathbf{H}$ expression in \eqref{eq:H_def}, and consequently in~\eqref{eq:Z_def}, thus, completing the proof.
\end{proof}
The beamspace of an example $12 \times 8$ MIMO channel matrix and its proposed block sparse representation in \eqref{eq:Z_def} via the previous Proposition 2 are illustrated in Fig.~\ref{fig:beamspace} for the case of ${L_p}=3$ channel propagation paths. It can be observed that the proposed beamspace forms a block sparse structure, according to which the non-zero values are concentrated along the diagonal.

\begin{figure}[t]
    \centering
    \includegraphics[scale=0.45]{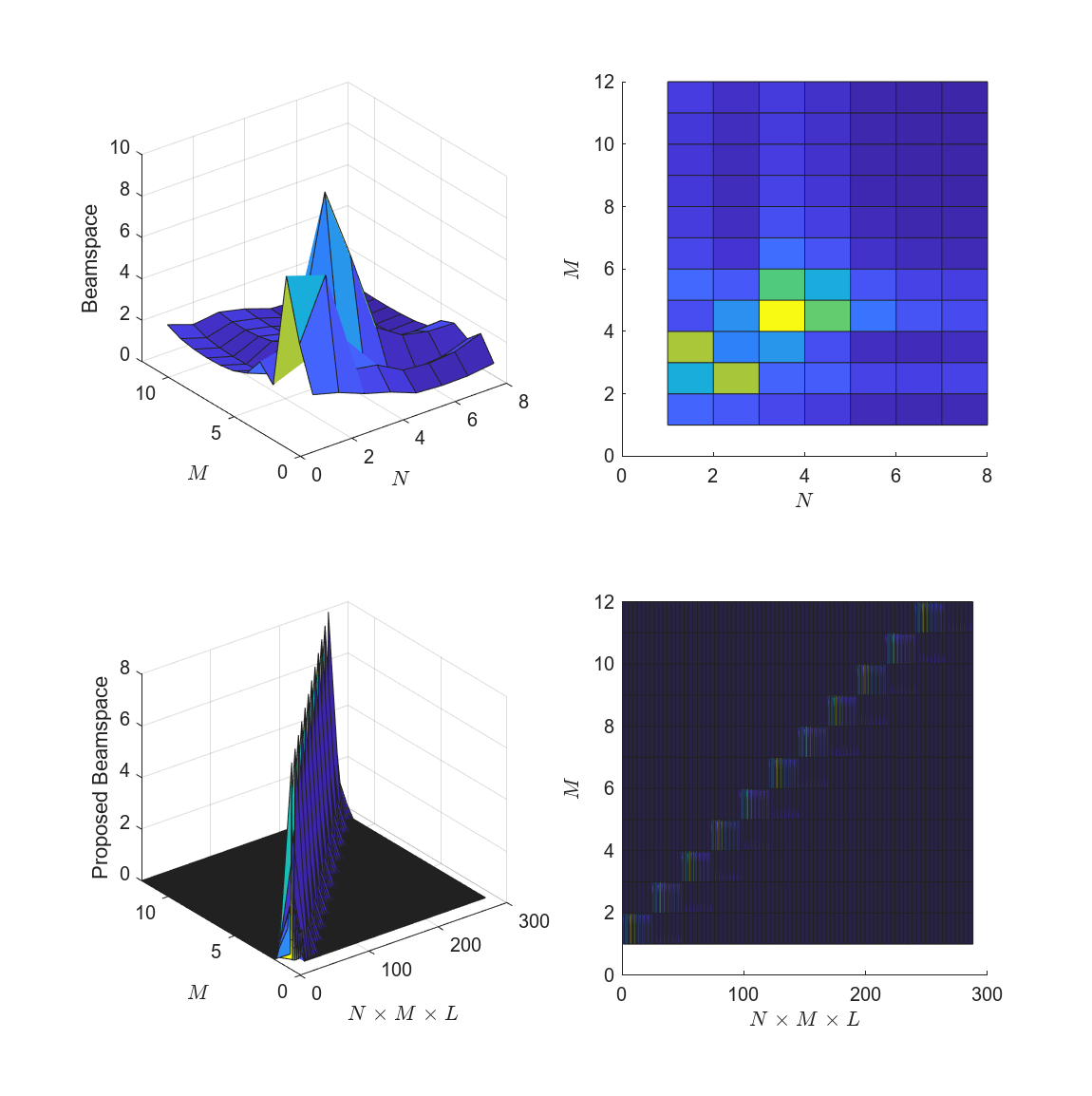}
    \caption{The beamspace $\Vert \mathbf{Z} \Vert_F^2$ of a $12 \times 8$ MIMO channel matrix (top) and that of its proposed block sparse representation in \eqref{eq:Z_def} via Proposition 2 (bottom) for the case of ${L_p}=3$ channel propagation paths.}
    \label{fig:beamspace}
\end{figure}

The inclusion of the channel sparsity property in the beamspace domain via Proposition~2 into $\mathcal{OP}$'s cost function permits us to re-express our channel estimation problem as:
\begin{align} 
    \mathcal{OP}_{\rm sp}:
     \min_{\mathbf{Z}, \mathbf{H}, \mathbf{E}} & \Vert \mathbf{Z} \Vert_1 + \Vert \mathbf{Y} - \mathbf{H} \mathbf{\Phi} \mathbf{E} 
    \Vert_{\rm F}^2 \,\,\text{s.t.}\,\, [\mathbf{E}]_{p,q} \in \{0,1\} \nonumber \\ \nonumber
    &\forall p=1,\ldots,MNLK \,\,\text{and}\,\, \forall q=1,\ldots,T,\\ \nonumber
    & \mathbf{H} \,\,\text{as in}\,\,\eqref{eq:H_def}\,\,\text{and}\,\,\mathbf{E} \,\,\text{as in}\,\,\eqref{eq:E_def}, \nonumber \\
    &\mathbf{H} = \mathbf{F}_1 \mathbf{Z} \mathbf{F}_2.
\end{align}
It is noteworthy that the introduction of matrix $\mathbf{Z}$ not only enhances the problem's solvability, via exploiting sparse optimization tools, but also transforms the problem into a formulation that allows for the efficient ADMM implementation~\cite{admm}, as it will be demonstrated in the sequel.

\subsection{Proposed Solution}
Let us first relax the binary constraint for $\mathbf{E}$, following the approach outlined in Algorithm 1. Then, $\mathcal{OP}_{\rm sp}$ becomes:
\begin{align}
    \mathcal{OP}_{\rm sp}':    \min_{\mathbf{Z}, \mathbf{H}, \mathbf{\tilde{e}}} & \Vert \mathbf{Z} \Vert_1 + \Vert \mathbf{\tilde{e}} \Vert_1 \nonumber \\
    & +  \frac{1}{2} \Vert \text{vec}(\mathbf{Y}) - (\mathbf{I} \otimes \mathbf{H}^{*} \mathbf{\Phi}) \text{vec}(\mathbf{I}_T \otimes \mathbf{\tilde{e}}) \Vert_2^2 \nonumber \\
    \text{s.t.}\,\,& \mathbf{H} = \mathbf{F}_1 \mathbf{Z} \mathbf{F}_2.\nonumber
\end{align}
The augmented Lagrangian for $\mathcal{OP}_{\rm sp}'$ is given by:
\begin{align}
    \mathcal{L}_\rho(\mathbf{Z}, \mathbf{H}, \mathbf{\tilde{e}}) = &\Vert \mathbf{Z} \Vert_1 + \Vert \mathbf{\tilde{e}} \Vert_1 \nonumber \\ &+\frac{1}{2} \Vert \text{vec}(\mathbf{Y}) - (\mathbf{I} \otimes \mathbf{H}^{*} \mathbf{\Phi}) \text{vec}(\mathbf{I}_T \otimes \mathbf{\tilde{e}}) \Vert_2^2 \nonumber \\
    &+ \langle \mathbf{C}, \mathbf{H} - \mathbf{F}_1 \mathbf{Z} \mathbf{F}_2  \rangle + \frac{\rho}{2} \Vert \mathbf{H} - \mathbf{F}_1 \mathbf{Z} \mathbf{F}_2 \Vert_{\rm F}^2.
\end{align}
The ADMM approach consists of the iterations~\cite{admm}:
     \begin{align}
     \mathbf{E}^{(j+1)} &= \arg \min_{\mathbf{E}} \mathcal{L}_\rho(\mathbf{Z}^{(j)}, \mathbf{H}^{(j)}, \mathbf{E}), \label{eq:admm_step1} \\
     \mathbf{Z}^{(j+1)} &= \arg \min_{\mathbf{Z}} \mathcal{L}_\rho(\mathbf{Z}, \mathbf{H}^{(j)}, \mathbf{E}^{(j+1)}), \label{eq:admm_step2} \\
     \mathbf{H}^{(j+1)} &= \arg \min_{\mathbf{H}} \mathcal{L}_\rho(\mathbf{Z}^{(j+1)}, \mathbf{H}, \mathbf{E}^{(j+1)}), \label{eq:admm_step3} \\
     \mathbf{C}^{(j+1)} &= \mathbf{C}^{(j)} + \rho (\mathbf{H}^{(j+1)} - \mathbf{F}_1 \mathbf{Z}^{(j+1)} \mathbf{F}_2), \label{eq:admm_dual}
    \end{align}   
Evidently, this iterative procedure necessitates the provision of initial values for the matrix $\mathbf{H}$, i.e.,  $\mathbf{H}^{(0)}$. We address this issue in a subsequent subsection. For the dual variable $\mathbf{C}^{(0)}$ needed in~\eqref{eq:admm_dual}, we initialize as the zeros' matrix.

\subsubsection{Derivation of $\mathbf{E}^{(j+1)}$}
To derive the matrix $\mathbf{E}$, we need to leverage its unique structure given by \eqref{eq:E_def}. Based on $\mathcal{OP}_{\rm sp}$, the optimization problem in \eqref{eq:admm_step1} can be expressed as follows:
\begin{equation}\label{eq:e_est}
    \min_{\mathbf{\tilde{e}}} \Vert \mathbf{\tilde{e}} \Vert_1 +\frac{1}{2} \Vert \text{vec}(\mathbf{Y}) - (\mathbf{I} \otimes \mathbf{H}^{*} \mathbf{\Phi}) \text{vec}(\mathbf{I}_T \otimes \mathbf{\tilde{e}}) \Vert_2^2,
\end{equation}
which can be easily solved as a standard LASSO problem. Then, a thresholding function $\text{thres}(\cdot)$ needs to applied at each element of the solution vector, yielding $\mathbf{e} \triangleq \text{thres}(\mathbf{\tilde{e}}) \in (0, 1)$ and $\Vert \mathbf{e} \Vert_0 = K_u \ll K$. 

\subsubsection{Derivation of $\mathbf{Z}^{(j+1)}$}
The optimization problem in~\eqref{eq:admm_step2} can be expressed as a standard LASSO problem. To do so, we first re-write its Lagrangian function with respect to $\mathbf{Z}$ as:
$$
\mathcal{L}_\rho(\mathbf{Z}) = \Vert \mathbf{Z} \Vert_1 + \langle \mathbf{C}^{(j)}, \mathbf{H}^{(j)} - \mathbf{F}_1 \mathbf{Z} \mathbf{F}_2 \rangle + \frac{\rho}{2} \Vert \mathbf{H}^{(j)} - \mathbf{F}_1 \mathbf{Z} \mathbf{F}_2 \Vert_{\rm F}^2.
$$
Adding the term $\frac{2}{\rho} \Vert \mathbf{C}^{(j)} \Vert_{\rm F}^2$ to the Lagrangian, we obtain the following equivalent problem:
\begin{equation}\label{eq:Z_prob}
    \min_{\mathbf{Z}} \Vert \mathbf{Z} \Vert_1 + \frac{\rho}{2} \Vert \frac{2}{\rho} \mathbf{C}^{(j)} + \mathbf{H}^{(j)} - \mathbf{F}_1 \mathbf{Z} \mathbf{F}_2 \Vert_{\rm F}^2.
\end{equation}
The estimation for $\mathbf{Z}^{(j+1)}$ in this problem can be achieved through various techniques addressing the LASSO problem. In this paper, we leverage CVX for its solution, while a comprehensive exploration of performance comparisons among different algorithms remains a subject for future investigation.  

\subsubsection{Derivation of $\mathbf{H}^{(j+1)}$}
Let us express the Lagrangian with respect to $\mathbf{H}$ for the problem in \eqref{eq:admm_step3} as follows:
\begin{align}
\mathcal{L}_\rho(\mathbf{H}) = &\Vert \mathbf{Y} - \mathbf{H}\mathbf{\Phi} \mathbf{E}^{(j+1)} \Vert_{\rm F}^2 + \langle\mathbf{C}^{(j)}, \mathbf{H}^{(j)} - \mathbf{F}_1 \mathbf{Z}^{(j+1)} \mathbf{F}_2 \rangle \nonumber \\
&+ \frac{\rho}{2} \Vert \mathbf{H}^{(j)} - \mathbf{F}_1 \mathbf{Z}^{(j+1)} \mathbf{F}_2  \Vert_{\rm F}^2.
\end{align}
To obtain the closed-form solution for this problem, we calculate the derivative of this Lagrangian, i.e.:
\begin{align}\label{eq:Lagrangian_H}
\frac{\partial \mathcal{L}_\rho (\mathbf{H})}{\mathbf{H}} = &\frac{\partial \Vert \mathbf{Y} \!-\! \mathbf{H}\mathbf{\Phi} \mathbf{E}^{(j+1)} \Vert_{\rm F}^2}{\partial \mathbf{H}} \!+\! \frac{\partial \langle\mathbf{C}^{(j)}, \mathbf{H} - \mathbf{F}_1 \mathbf{Z}^{(j+1)} \mathbf{F}_2 \rangle}{\partial \mathbf{H}}\nonumber \\
& + \frac{\rho}{2}\frac{\partial \Vert \mathbf{H} - \mathbf{F}_1 \mathbf{Z}^{(j+1)} \mathbf{F}_2 \Vert_{\rm F}^2}{\partial \mathbf{H}}.
\end{align}
Let us calculate each term of \eqref{eq:Lagrangian_H} separately. The first term can be expressed as follows:
$$
\frac{\partial \Vert \mathbf{Y} - \mathbf{H}\mathbf{\Phi} \mathbf{E}^{(j+1)} \Vert_{\rm F}^2}{\partial \mathbf{H}} = -2\left(\mathbf{Y} - \mathbf{H} \mathbf{\Phi} \mathbf{E}^{(j+1)}\right).
$$
Then, the second term is given by:
\begin{align}
& \frac{\partial \langle\mathbf{C}^{(j)}, \mathbf{H} - \mathbf{F}_1 \mathbf{Z}^{(j+1)} \mathbf{F}_2\rangle}{\partial \mathbf{H}} = \frac{\partial (\mathbf{C}^{(j)})^{\rm H}(\mathbf{H} - \mathbf{F}_1 \mathbf{Z}^{(j+1)} \mathbf{F}_2)}{\partial \mathbf{H}} \nonumber \\ &+ \frac{\partial (\mathbf{H} - \mathbf{F}_1 \mathbf{Z}^{(j+1)} \mathbf{F}_2)^{\rm H} \mathbf{C}^{(j)}}{\partial \mathbf{H}} = (\mathbf{C}^{(j)})^{\rm H} + \mathbf{C}^{(j)}
\end{align}
The third term is calculated as:
\begin{align}
\frac{\partial \Vert \mathbf{H} \!-\! \mathbf{F}_1 \mathbf{Z}^{(j+1)} \mathbf{F}_2 \Vert_{\rm F}^2}{\partial \mathbf{H}} = -2 \left(\mathbf{H} - \mathbf{F}_1 \mathbf{Z}^{(j+1)} \mathbf{F}_2\right).
\end{align}
Putting all above together, the closed-form solution of \eqref{eq:Lagrangian_H} is given by the following equation:
\begin{align}
&\frac{\partial \mathcal{L}_\rho (\mathbf{H})}{\partial \mathbf{H}} = -2\left(\mathbf{Y} - \mathbf{H} \mathbf{\Phi} \mathbf{E}^{(j+1)}\right) + (\mathbf{C}^{(j)})^{\rm H} + \mathbf{C}^{(j)} \nonumber \\ & - 2 \left(\mathbf{H} - \mathbf{F}_1 \mathbf{Z}^{(j+1)} \mathbf{F}_2\right)  = \mathbf{0}_{M \times MNL} \label{eq:H_est},
\end{align}
which can be easily solved over the unknown matrix $\mathbf{H}$.

The ADMM steps solving $\mathcal{OP}_{\rm sp}$ are summarized in Algorithm~\ref{alg:proposed2}. Therein, the initialization process described in Steps $1$ and $2$ for computing $\mathbf{H}^{(0)}$ will be described in the next subsection. To solve the problems included in Steps $4$ and $6$ we employ the CVX tool~\cite{cvx}.

\begin{algorithm}[t]
\caption{Proposed XL MIMO Estimation}
\label{alg:proposed2}
\begin{algorithmic}[1]
\REQUIRE $\mathbf{\Phi}$, $\mathbf{Y}$, $x_\text{BS}$, $x_\text{UE}$, and $I_\text{max}$.
\ENSURE $\mathbf{H}^{(I_\text{max})}$ and $\mathbf{E}^{(I_\text{max})}$.
\STATE Obtain TX-RX AoA and AoD from \eqref{eq:theta_rx} and \eqref{eq:theta_tx}.
\STATE Initialize the channel matrix $\mathbf{H}^{(0)}$ using \eqref{eq:steering_matrix_init}.
\FOR{$j=1,\ldots,I_\text{max}$} 
\STATE Solve $\eqref{eq:e_est}$ and compute $\mathbf{e}^{(j+1)}$.
\STATE Compute $\mathbf{E}^{(j+1)} = \mathbf{I}_T \otimes \mathbf{e}^{(j+1)}$.
\STATE Solve \eqref{eq:Z_prob} to obtain $\mathbf{Z}^{(j+1)}$.
\STATE Calculate $\mathbf{H}^{(j+1)}$ via \eqref{eq:H_est}.
\STATE Update the dual variable via \eqref{eq:admm_dual}.
\ENDFOR
\end{algorithmic}
\end{algorithm}

\subsection{Initialization based on UE Position Information}
UE position knowledge has been widely exploited to enhance XL MIMO channel estimation and reduce the number of training symbols~\cite{taranto2014location}.  In our case, to recover $\mathbf{E}$ in the proposed iterative algorithm, knowledge of the initial instantaneous channel matrix $\mathbf{H}^{(0)}$ is required. By representing the UE and BS positions on the 2D-plane as $(x_{\rm UE}, y_{\rm UE})$ and $(x_{\rm BS}, y_{\rm BS})$, respectively, the physical AoA and AoD are given by:
\begin{align}
    \hat{\vartheta}_{\text{RX}} = \arcsin \frac{x_{\rm UE}}{\sqrt{(x_\text{BS} - x_\text{UE})^2 + (y_\text{BS} - y_\text{UE})^2}},\label{eq:theta_rx}\\
    \hat{\vartheta}_{\text{TX}} = \arcsin \frac{x_{\rm BS}}{\sqrt{(x_\text{BS} - x_\text{UE})^2 + (y_\text{BS} - y_\text{UE})^2}}, \label{eq:theta_tx}
\end{align}
thus, the normalized versions are respectively $\hat{\theta}_\text{RX} = \frac{\Delta_c \sin \hat{\vartheta}_\text{RX}}{\lambda_c}$ and $\hat{\theta}_\text{TX} = \frac{\Delta_c \sin \hat{\vartheta}_1}{\lambda_c}$. Therefore, the respective steering vectors for the LoS component are computed as follows:
\begin{equation}
    \mathbf{a}_{\rm RX}(\hat{\theta}_\text{RX}) = \frac{1}{\sqrt{M}}\left[ 1, e^{-j2 \pi \hat{\theta}_{\rm RX}}, \ldots, e^{-j2 \pi (N-1) \hat{\theta}_{\rm RX}}\right]^{\rm T},
\end{equation}
\begin{equation}
    \mathbf{a}_{\rm TX}(\hat{\theta}_{\rm TX}) = \frac{1}{\sqrt{N}}\left[ 1, e^{-j2 \pi \hat{\theta}_{\rm TX}}, \ldots, e^{-j2 \pi (N-1) \hat{\theta}_{\rm TX}}\right]^{\rm T}.
\end{equation}
Finally, the instantaneous channel vectors for each RX antenna antenna at the UE can be approximated as follows:
\begin{equation}\label{eq:steering_matrix_init}
    \mathbf{h}_m^{(0)} \triangleq [\mathbf{a}_{\rm RX}(\hat{\theta}_{\rm RX})]_m \mathbf{a}^{\rm H}_{\rm TX}(\hat{\theta}_{\rm TX}).
\end{equation}
While this approach offers an approximation for the $m$-th channel, since the instantaneous gain cannot be retrieved in this manner, recall that the angular matrix $\mathbf{E}$ encodes the angular shift caused by the beam-squint effect. Therefore, it remains unaffected by signal amplitude variations.
  
\subsection{Complexity Analysis} 
We now elaborate into the computational complexity of the proposed XL MIMO channel estimation technique, as summarized in Algorithm~\ref{alg:proposed1}. The ensuing calculations are conducted over a span of $I_\text{max}$ iterations, ensuring the convergence of the ADMM technique.

\subsubsection{Calculation of the binary matrix $\mathbf{E}$} The $MNLK \times 1$ binary vector $\mathbf{e}^{(j+1)}$ and the $MNLKT \times T$ binary matrix $\mathbf{E}^{(j+1)}$ are obtained in Lines $4$ and $5$. To address the integer programming problem in~\eqref{eq:e_est}, we approximate it using box constraints, and then, utilize the CVX package to solve the resulting optimization. The latter constitutes the major computational cost, whose complexity can be upper bounded by $\mathcal{O}(MNKL)$. 

As it will be shown numerically in the next section (specifically in Section~\ref{sec:combined}), vector 
$\mathbf{e}$ is expected to be sparse for the considered dual-wideband fading conditions at THz frequencies. Interestingly, this sparsity can be leveraged to reduce the complexity of the $\mathbf{E}^{(j)}$ computation at each $j$-th algorithmic iteration. We propose to deploy a heuristic approach based on the orthogonal matching pursuit (OMP)~\cite{AVW22a} to replace \eqref{eq:e_est} and efficiently solve the $\ell_1$ optimization problem:
\begin{equation*}
    \min_{\mathbf{\tilde{e}}} \Vert \mathbf{\tilde{e}} \Vert_1 + \frac{1}{2} \Vert \text{vec}(\mathbf{Y}) - (\mathbf{I} \otimes \mathbf{H}^{(j)} \mathbf{\Phi}) \text{vec}(\mathbf{I}_T \otimes \mathbf{\tilde{e}}) \Vert_2^2.
\end{equation*}
By thresholding this problem solution vector $\mathbf{\tilde{e}}$ to obtain  $\mathbf{\hat{e}}$, an estimation for the binary matrix is obtained as $\mathbf{E}^{(j)} = \mathbf{I}_T \otimes \mathbf{\hat{e}}$. This procedures entails a much lower complexity order, namely $\mathcal{O}(MNK_uL)$ with $K_u \ll K$.

\subsubsection{Calculation of the matrix $\mathbf{Z}$} The beamspace matrix $\mathbf{Z}$ is computed in Line $6$ of Algorithm~\ref{alg:proposed1}. The computational complexity for solving the LASSO problem in~\eqref{eq:Z_prob} depends on the chosen algorithm. For example, algorithms like coordinate descent and proximal gradient descent have different complexities. As a general rule, the overall solution complexity scales with the number of features (e.g., the rank ${L_p}$ of the channel matrix) and the number of samples (i.e., the training length $T$). Therefore, a rough estimate of the complexity for solving \eqref{eq:Z_prob} is in the order of $\mathcal{O}({L_p} T)$.

\subsubsection{Calculation of the matrix $\mathbf{H}$} Line $7$ computes the $M \times NLK$ channel matrix. While equation \eqref{eq:H_est} needs to be solved, most matrices involved therein typically exhibit sparse structures. This sparsity significantly reduces the computational cost of matrix manipulations compared to their full-dimensional counterparts.

\section{Numerical Results and Discussion}

In this section, we evaluate the performance of the proposed technique via computer simulation results using MATLAB$^{\textrm{TM}}$.

\subsection{System Setup}
\begin{table}[t]
\caption{Default setting of the simulation parameters.}
\centering
\begin{tabular}{ll} 
 \textbf{Setting} & \textbf{Value} \\
\toprule
Monte-Carlo realizations & $R=100$ \\
Carrier Frequency & $f_c=150$ GHz \\ 
System Bandwidth & $B=10$ GHz \\
Channel Coherence time & $100$ nsec \\
Number of TX Antennas &  $64 \le N \le 256$  \\
Number of RX Antennas &  $64 \le M \le 256$ \\
BS transmit power  & $P_t=10$ dBm \\
Number of multipath components & ${L_p}=3$ \\
LOS pathloss exponent & $\xi_1=2$ \\
nLOS pathloss exponent & $\xi_\ell=3$, for $\ell=2,3,\ldots, {L_p}$ \\
TX and RX distance & $1 \le d_\text{BS-UE} \le 10$ m \\
\bottomrule
\end{tabular}
\label{tab:simul}
\end{table}

We assume that the considered point-to-point THz communication system takes place for distances larger than the Fraunhofer distance \cite{Balanis-2012-antenna}. According to the SC modulation under consideration, the TX communicates with the UE via data-carrying frames, where each frame is composed by $T$ time instances allocated for the training symbols. Thus, for $t=1,2,\ldots,T$, the TX transmits the training symbols $\bar{q}_n(t) \sim \mathcal{CN}(0,1)$. Moreover, each time frame has been considered as a new Monte-Carlo realization for all involved random variables (i.e., thermal noise and complex channel gains). The default channel and system parameter settings are included in Table \ref{tab:simul}.

To evaluate the performance of the proposed technique in Algorithm~2 and compare with other benchmarks, we have considered the normalized mean square error (NMSE) metric for the channel estimation, which is defined as follows:
\begin{equation}
    \text{NMSE} \triangleq \sum_{r=1}^R \sum_{m=1}^M \frac{\Vert \mathbf{h}_m - \hat{\mathbf{h}}_m \Vert}{\Vert \mathbf{h}_m \Vert},
\end{equation}
where $\hat{\mathbf{h}}_m$ represents the estimated vector of the $m$-th RX antenna and $R$ indicates the total number of Monte-Carlo realizations. The evaluation took place considering various scenarios with different Signal-to-Noise Ratios (SNRs). In addition, we include a benchmark comparison with the idealized solution described in Algorithm~1, which is constant over iterations.

\subsection{Convergence Analysis}
\begin{figure}[t]
    \centering
    \includegraphics[scale=0.5]{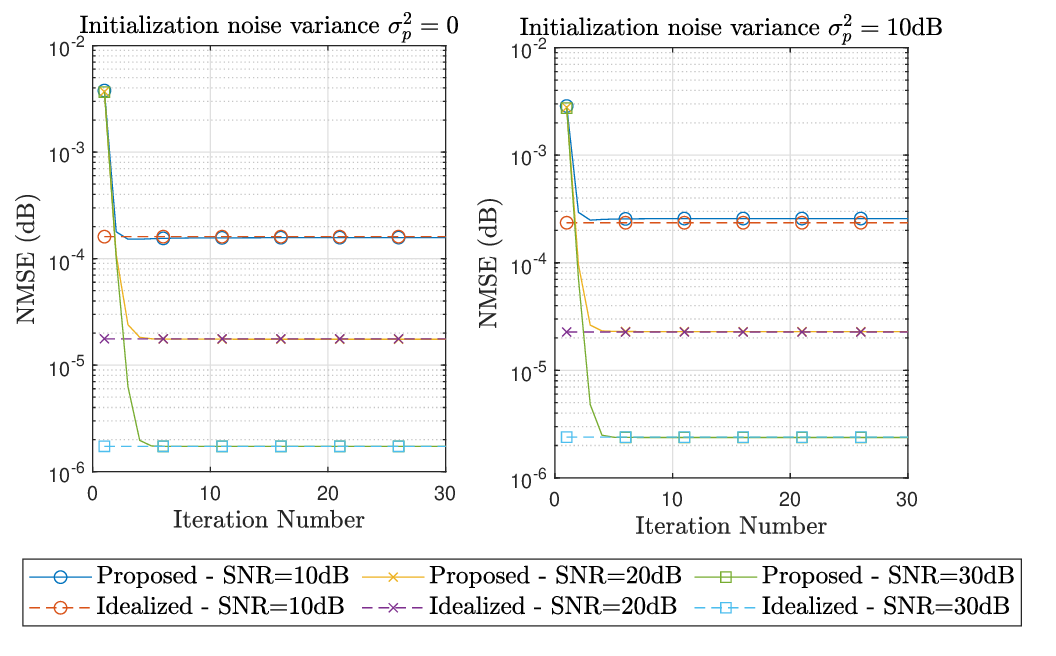}
    \caption{Convergence of the the proposed channel estimation algorithm with $N=M=256$, $T=3N=768$, and ${L_p}=6$.}
    \label{fig:convergence}
\end{figure}
We commence with a $256 \times 256$ MIMO system having a distance between TX and RX equal to $d_\text{TX-RX}=1$m and operating under a subTHz channel with ${L_p}=3$ propagation paths. In Fig.~\ref{fig:convergence} (a), we have set the number of training symbols as $T=3N=768$ to investigate the convergence speed of the proposed channel estimation algorithm, while the values for the other simulation parameters are given in Table~\ref{tab:simul}. The initialization of the proposed algorithm relied on an AoA information with $\sigma_p^2=0$. In contrast, the idealized approach has the advantage of being initialized with perfect knowledge of the channel or the binary beam-squint vector. It can be observed that the curves exhibit a swift convergence across all simulated SNRs, necessitating only a few iterations of the ADMM algorithm with a step size of $\rho=6$.

The position information which is used for initialization via \eqref{eq:steering_matrix_init} is assumed to be precise and devoid of any errors. However, in practical situations, it is anticipated that this information may be subject to noise due to various factors, including inaccuracies in position estimation and other disturbances, such as phase noise. To provide a more accurate representation of real-world phenomena, we adopt the following general noisy model for the channel initialization:
\begin{equation}\label{eq:steering_matrix_init_noisy}
    \mathbf{\hat{H}}^{(0)} = \mathbf{H}^{(0)} + \mathbf{W},
\end{equation}
where $[\mathbf{W}]_{i,j} \sim \mathcal{CN}(0, \sigma_p^2)$. In Fig. \ref{fig:convergence}(b), we include the convergence curves keeping the same parameter values as previously, but for the case of noisy positions, following the expression \eqref{eq:steering_matrix_init_noisy} with the setting $\sigma_p^2=10$dB. As shown, despite a slightly higher error bounds, the proposed algorithm still successfully reaches the idealized values.

\subsection{Combined Time Delay Profile}\label{sec:combined}
We now examine how the TX/RX spatial wideband effects, multipath propagation, and molecular absorption combinedly contribute to delays in the received signal. These delays are represented by the binary vector $\mathbf{e}  \in \{0,1\}^{MNLK \times 1}$ defined in \eqref{eq:E_def} within Proposition 1, where a unity value indicates a delay caused by any of the latter factors. Recall that the number of the non-zero (unity) values of $\mathbf{e}$ equals to $K_u$, while the maximum delay length is $K T_s$ in seconds. Considering an example setting with $f_c=150$ GHz and $N=M=32$, the vector $\mathbf{e}_{m,n}$ between the $n$-th TX and the $m$-th RX antenna elements will be composed by one unity value due to the beam-squint of the antennas and ${L_p}$ unities due to the combined contribution of multipath propagation and molecular absorption. This vector $\mathbf{e}_{m,n}$ is depicted over its index $k$, with $k=1,2,\ldots, K$, in Fig.~\ref{fig:plot_e}.

\begin{figure}[t]
    \centering
    \includegraphics[scale=0.45]{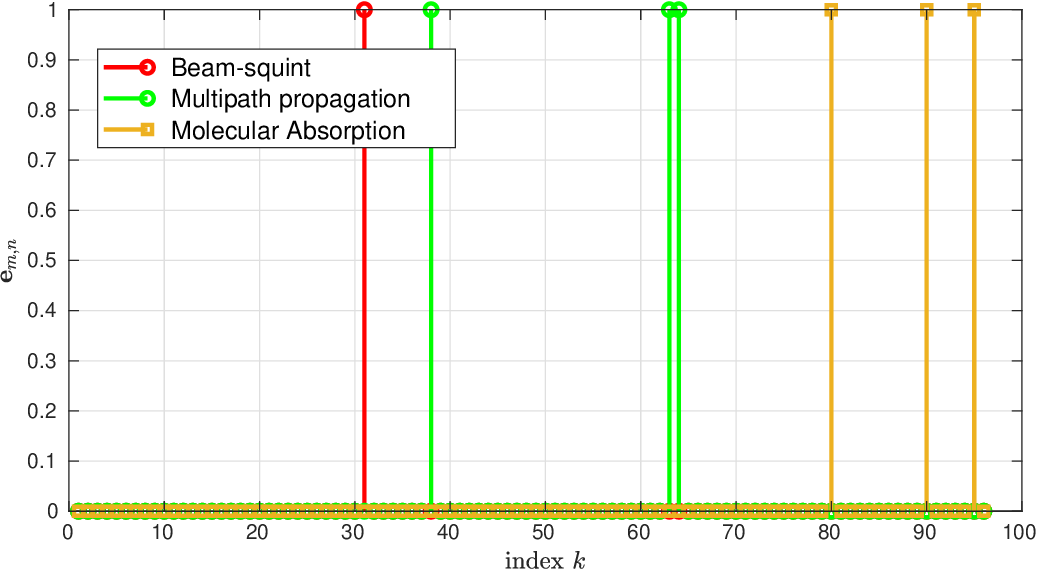}
    \caption{Simplified example of the combined time delay profile between the $n$-th TX and the $m$-th RX antenna elements for a $32 \times 32$ MIMO system resulting from the TX/RX spatial wideband effects, multipath propagation, and molecular absorption.}
    \label{fig:plot_e}
\end{figure}

\begin{figure}[t]
    \centering
    \includegraphics[scale=0.45]{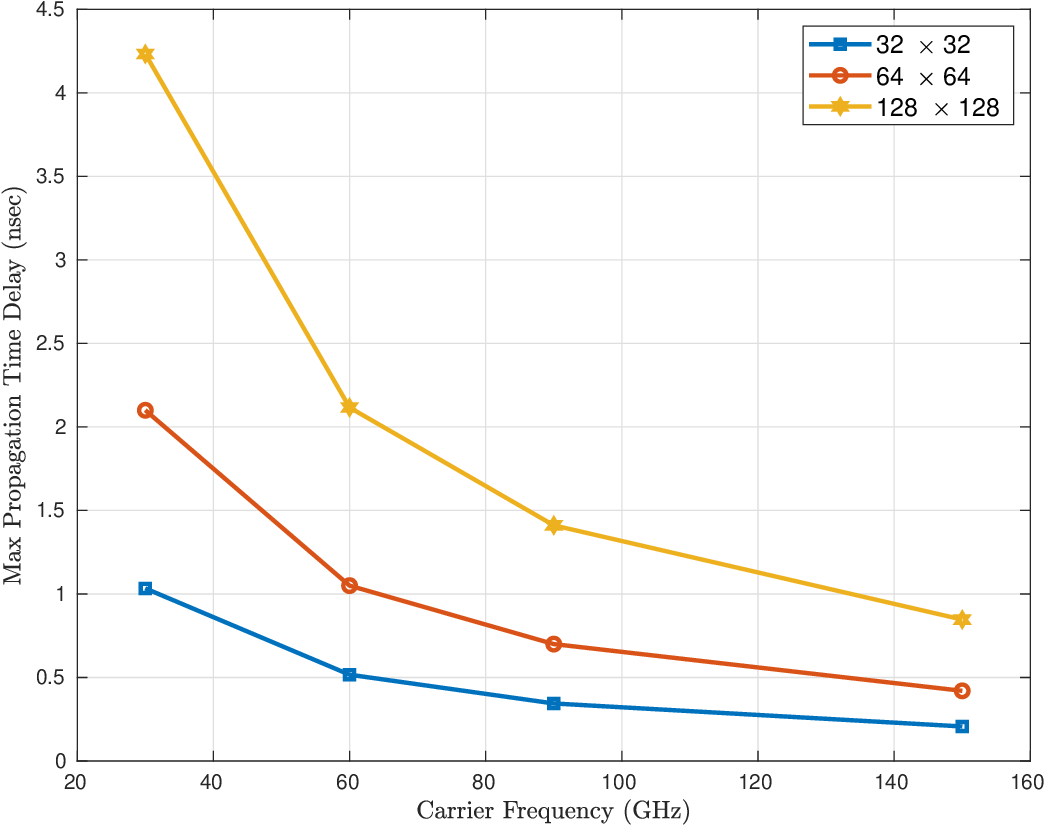}
    \caption{Time delay due to the spatial wideband effects versus the carrier frequency in GHz.}
    \label{fig:beamsquint_time_delay}
\end{figure}

\subsection{Channel Estimation Performance Evaluation}

Let us now investigate the NMSE performance of channel estimation as a function of the SNR. For comparisons, we have simulated the performance of the following benchmarks:
\begin{itemize}
\item \textit{Least-Squares}: This approach ignores the beam-squint effect and provides the following solution for channel estimation:
\begin{equation}
    \mathbf{h}_\text{LS} = \mathbf{\Phi}^\dagger \mathbf{y}.
\end{equation}
\item \textit{OMP}: This compressive sensing method solves the following system:
\begin{equation}
    \min_{\mathbf{z}} \Vert \mathbf{z} \Vert_1 + \Vert \mathbf{y} - \mathbf{\Phi} \mathbf{F} \mathbf{z} \Vert_2^2,
\end{equation}

\item \textit{Idealized (Algorithm 1)}: The two problems of beam-squint, as formulated in $\mathcal{OP}_1$, and channel estimation, expressed via $\mathcal{OP}_2$, have been solved independently (via the CVX package \cite{cvx}). When solving for $\mathbf{H}$, perfect knowledge of $\mathbf{E}$ was assumed. Respectively, when solving for $\mathbf{E}$, perfect knowledge for the channel $\mathbf{H}$ was assumed.
\end{itemize}

Channel estimation mainly relies on receiving known training pilots. In SC systems, pilots are typically placed at the beginning of each block, assuming a constant channel within that block. In contrast, MC systems, like Orthogonal Frequency Division Multiplexing (OFDM), often dedicate specific subcarriers for pilot transmission. For both SC and MC systems, we consider transmission within a frame of duration $T$ and bandwidth $B=M \Delta f$. The time-frequency (TF) domain was discretized into a lattice by sampling time and frequency at integer multiples of $1/B$ and $\Delta f$, respectively. When the frame duration is $30$ nsec, the bandwidth is $B=30$ GHz, and then, the length of the frame is $T=900$ time instances.

As previously discussed for MC systems, due to the beam-squint effect, each subcarrier experiences a different channel, thus, channel estimation has to performed for each carrier separately. This increases significantly the complexity and computational demands as compared to SC methods. A naive approach for channel estimation in wideband massive MIMO systems is to insert training pilots across all time instances and frequency bins. However, this significantly increases the pilot overhead within each frame. An alternative strategy is to divide the large antenna array into smaller subarrays. While each subarray remains susceptible to beam squint (a frequency-dependent beam direction), techniques like true time delays (TTD) can be employed to combine the signals from these subarrays. This allows for the use of OFDM on each individual subarray, effectively reducing the pilot overhead compared to the naive approach. However, the employment of TTDs may also introduce angular resolution reduction compared to phase-shifters, i.e., phase noise.

Figure~\ref{fig:mseVsnr} illustrates the comparison for a system with  $N=M=64$, $T=3N=192$, and ${L_p}=3$. The results indicate that both the proposed techniques in Algorithms~\ref{alg:proposed1} and~\ref{alg:proposed2} exhibit almost identical performance with the LS approach and the Idealized SC method. It is also shown that the OMP method exhibits the worst performance due to the grid discretization errors for the beamspace. In addition, it is observed that the MC based approach is not able to attain the same NMSE level as its SC counterpart. This is attributed to the limited number of training symbols. 
In Fig.~\ref{fig:mseVtraininglength}, we illustrate the NMSE performance as it varies with the training length $T$ for a fixed SNR value at $30$dB and a $64 \times 64$ system. The results verify that the proposed algorithms are able to achieve the idealized SC performance for training lengths over $T>N=64$.

\begin{figure}[t]
    \centering
    \includegraphics[scale=0.6]{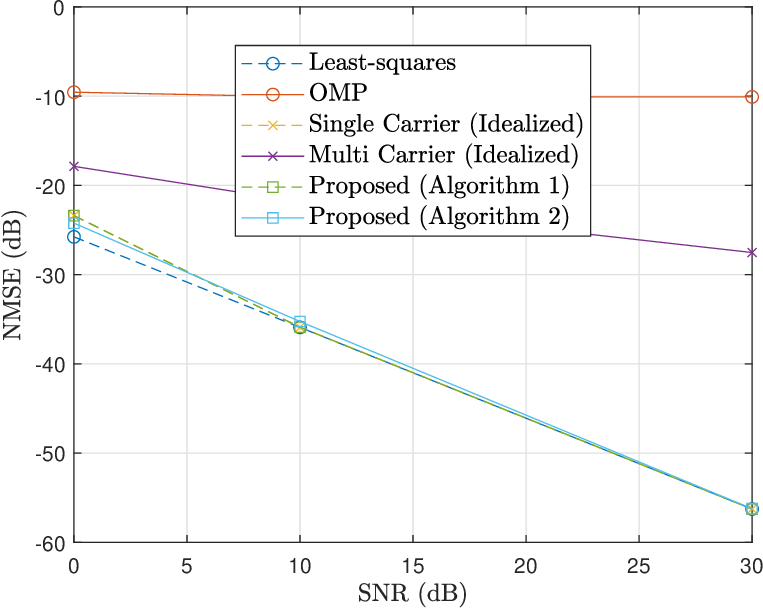}
    \caption{Channel estimation performance versus the SNR for a $64 \times 64$ MIMO system with ${L_p}=3$, $T=3N=192$, and $\sigma_p^2=10$ dB.}
    \label{fig:mseVsnr}
\end{figure}

\begin{figure}[t]
    \centering
    \includegraphics[scale=0.6]{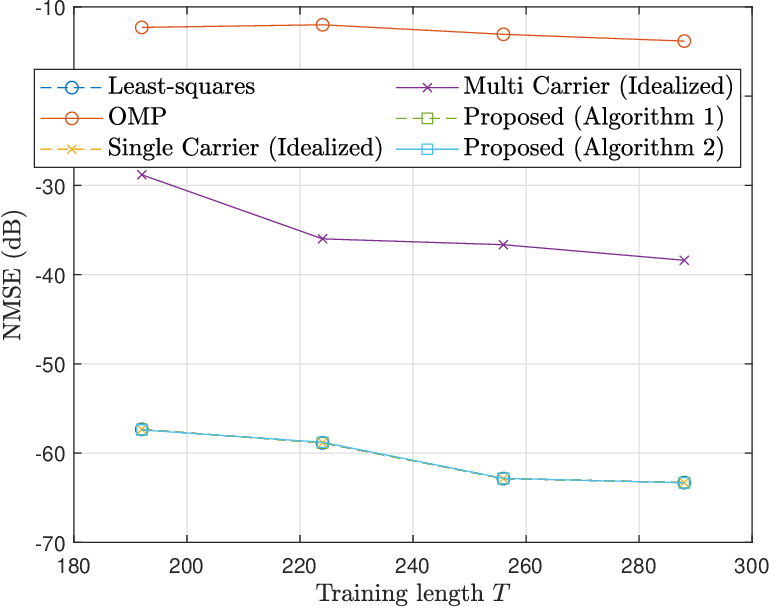}
    \caption{Channel estimation performance versus the training length $T$ with ${L_p}=3$ and SNR equal to $30$ dB.}
    \label{fig:mseVtraininglength}
\end{figure}

\section{Conclusions}
In this paper, we addressed XL MIMO channel estimation in the THz frequency band considering array-wide propagation delays causing frequency-selective beam squint. Traditional frequency modulation exhibits high peak-to-average power ratios, exacerbated by the low THz transmit powers. To confront with this issue, we presented a novel time-domain SC-based estimation approach, treating beam squint through sparse vector recovery via optimization. Our technique deployed alternating minimization to jointly handle the beam-squint effect and the MIMO channel sparsity. 
Given the inherent complexity of the considered non-linear XL MIMO estimation problem, our proposed technique leveraged the potential availability of position information of the user equipment at the base station to enhance the accuracy of the estimation process. The robustness of the proposed estimation technique in the presence of deviations from the true user position, resulting in erroneous partial composition of the LOS component of the unknown XL MIMO matrix, was thoroughly investigated. The presented performance evaluations showcased that the proposed XL MIMO estimation scheme exhibits superior
performance than conventional SC- and MC-based techniques,
approaching the idealized lower bound.



\bibliographystyle{ieeetr}
\bibliography{references}

\end{document}